\documentclass[a4paper,USenglish,cleveref,autoref,thm-restate]{lipics-v2021}
\pdfoutput=1
\newcommand{\ldist}{\mathsf {dist}}

\theoremstyle{plain}

\usepackage{microtype}
\usepackage{cite}
\usepackage{framed}
\usepackage[framemethod=tikz]{mdframed}
\usepackage{appendix}
\usepackage{graphicx}
\usepackage{color}
\usepackage{varwidth}
\usepackage{wrapfig}
\usepackage[boxed,noend,linesnumbered]{algorithm2e}
\usepackage[noend]{algpseudocode}
\usepackage[textsize=tiny]{todonotes}
\usepackage{array}
\usepackage[normalem]{ulem}
\usepackage{dsfont}

\usepackage{xspace}
\usepackage{xifthen}

\definecolor{darkgreen}{rgb}{0,0.5,0}
\definecolor{darkblue}{rgb}{0,0,0.8}
\usepackage{hyperref}
\hypersetup{
    unicode=false,          % non-Latin characters in Acrobat’s bookmarks
    colorlinks=true,        % false: boxed links; true: colored links
    linkcolor=darkblue,          % color of internal links (change box color with linkbordercolor)
    citecolor=darkgreen,        % color of links to bibliography
    filecolor=magenta,      % color of file links
    urlcolor=cyan           % color of external links
}
\RequirePackage[]{silence}
\newcommand{\ignore}[1]{}

\algnewcommand\algorithmicswitch{\textbf{switch}}
\algnewcommand\algorithmiccase{\textbf{case}}

\algdef{SE}[SWITCH]{Switch}{EndSwitch}[1]{\algorithmicswitch\ #1\ \algorithmicdo}{\algorithmicend\ \algorithmicswitch}%
\algdef{SE}[CASE]{Case}{EndCase}[1]{\algorithmiccase\ #1}{\algorithmicend\ \algorithmiccase}%
\algtext*{EndSwitch}%
\algtext*{EndCase}%
\newcommand{\LOCAL}{\ensuremath{\mathsf{LOCAL}}\xspace}

\newcommand{\eps}{\varepsilon}
\renewcommand{\epsilon}{\varepsilon}
\newcommand{\poly}{\operatorname{\text{{\rm poly}}}}

\newcommand{\set}[1]{\left\{#1\right\}}

\newcommand{\hide}[1]{}

\newcommand{\FullOrShort}{short}

\ifthenelse{\equal{\FullOrShort}{full}}{
	
  \newcommand{\fullOnly}[1]{#1}
  \newcommand{\shortOnly}[1]{}

}{
  
  \newcommand{\shortOnly}[1]{#1}
  \newcommand{\fullOnly}[1]{}

}

\renewcommand{\phi}{\varphi}

\title{Improved Distributed Fractional Coloring Algorithms}

\author{Alkida Balliu}{Univeristy of Freiburg}{alkida.balliu@cs.uni-freiburg.de}{}{}
\author{Fabian Kuhn}{Univeristy of Freiburg}{kuhn@cs.uni-freiburg.de}{}{}
\author{Dennis Olivetti}{Univeristy of Freiburg}{dennis.olivetti@cs.uni-freiburg.de}{}{}

\authorrunning{A.\ Balliu \and F.\ Kuhn \and D.\ Olivetti} 

\Copyright{Alkida Balliu, Fabian Kuhn, and Dennis Olivetti}



\nolinenumbers %uncomment to disable line numbering

\hideLIPIcs  %uncomment to remove references to LIPIcs series (logo, DOI, ...), e.g. when preparing a pre-final version to be uploaded to arXiv or another public repository

\begin{document}

\maketitle

%%
%% Abstract
%%
\begin{abstract}
  We prove new bounds on the distributed fractional coloring problem in the \LOCAL model. A fractional $c$-coloring of a graph $G=(V,E)$ is a fractional covering of the nodes of $G$ with independent sets such that each independent set $I$ of $G$ is assigned a fractional value $\lambda_I\in[0,1]$. The total value of all independent sets of $G$ is at most $c$, and for each node $v\in V$, the total value of all independent sets containing $v$ is at least $1$. Equivalently, fractional $c$-colorings can also be understood as multicolorings as follows. For some natural numbers $p$ and $q$ such that $p/q\leq c$, each node $v$ is assigned a set of at least $q$ colors from $\set{1,\dots,p}$ such that adjacent nodes are assigned disjoint sets of colors. The minimum $c$ for which a fractional $c$-coloring of a graph $G$ exists is called the fractional chromatic number $\chi_f(G)$ of $G$.

  Recently, [Bousquet, Esperet, and Pirot; SIROCCO '21] showed that for any constant $\eps>0$, a fractional $(\Delta+\eps)$-coloring can be computed in $\Delta^{O(\Delta)} + O(\Delta\cdot\log^* n)$ rounds. We show that such a coloring can be computed in only $O(\log^2 \Delta)$ rounds, without any dependency on $n$.%$O(\sqrt{\Delta}\poly \log\Delta + \log^* n)$ rounds.

  We further show that in $O\big(\frac{\log n}{\eps}\big)$ rounds, it is possible to compute a fractional $(1+\eps)\chi_f(G)$-coloring, even if the fractional chromatic number $\chi_f(G)$ is not known. That is, the fractional coloring problem can be approximated arbitrarily well by an efficient algorithm in the \LOCAL model. For the standard coloring problem, it is only known that an $O\big(\frac{\log n}{\log\log n}\big)$-approximation can be computed in polylogarithmic time in the \LOCAL model. We also show that our distributed fractional coloring approximation algorithm is best possible. We show that in trees, which have fractional chromatic number $2$, computing a fractional $(2+\eps)$-coloring requires at least $\Omega\big(\frac{\log n}{\eps}\big)$ rounds.

    We finally study fractional colorings of regular grids. In [Bousquet, Esperet, and Pirot; SIROCCO '21], it is shown that in regular grids of bounded dimension, a fractional $(2+\eps)$-coloring can be computed in time $O(\log^* n)$. We show that such a coloring can even be computed in $O(1)$ rounds in the \LOCAL model.
\end{abstract}

\section{Introduction \& Related Work}

The distributed graph coloring problem is at the heart of the area of distributed graph algorithms and it is one of the prototypical problems to study distributed symmetry breaking. Already in \cite{linial1987LOCAL}, Linial showed that deterministically coloring a ring network with $O(1)$ colors and thus more generally coloring $n$-node graphs of maximum degree $\Delta$ with a number of colors that only depends on $\Delta$ requires $\Omega(\log^* n)$ rounds. In \cite{naor91}, Naor extended this lower bound to randomized algorithms. Subsequently, over the last three decades, the distributed coloring problem has been studied intensively and we now have a quite good understanding of the complexity of the problem. Mostly, researchers focused on the problem of computing a coloring with $\Delta+1$ colors, i.e., with the number of colors that can be obtained by a simple sequential greedy algorithm. In light of the $\Omega(\log^* n)$ lower bounds of \cite{linial1987LOCAL,naor91}, there is a long line of research to (deterministically) solve $(\Delta+1)$-coloring in time $f(\Delta)+O(\log^* n)$ for some function $f$ (see, e.g., \cite{goldberg88,szegedy93,Kuhn2006On,BEK15,barenboim2016deterministic,fraigniaud16,BEG18,MausTonoyan20}), where the current best bound of $O(\sqrt{\Delta\log\Delta}+\log^* n)$ was proven in \cite{fraigniaud16,BEG18,MausTonoyan20}. The complexity of distributed $(\Delta+1)$-coloring has also been studied as a function of $n$. The randomized complexity has been known to be $O(\log n)$ since the 1980s~\cite{luby86,linial1987LOCAL,alon86,Johansson99} and it has recently been improved to $O(\log^3\log n)$~\cite{barenboim_symmbreaking,harris2016coloring,chang2018optimal,GhaffariKuhn21} and even to $O(\log^* n)$ for graphs of maximum degree $\Delta\geq\log^{2+\eps}n$~\cite{HNT21}. For a long time, the best deterministic $(\Delta+1)$-coloring algorithms had a complexity of $2^{O(\sqrt{\log n})}$~\cite{awerbuch89,panconesi-srinivasan,kuhn2020faster} and it was only recently shown in a breakthrough paper by Rozho\v{n} and Ghaffari~\cite{rozhonghaffari20} that the distributed $(\Delta+1)$-coloring problem (and many other distributed graph problems) can be solved in time $\poly\log n$ deterministically. Subsequently, the deterministic complexity of the distributed $(\Delta+1)$-coloring problem has even been improved to $O(\log^2\Delta \cdot\log n)$~\cite{GhaffariKuhn21}.

While much of the existing work on distributed vertex coloring is on the $(\Delta+1)$-coloring problem, it is of course also relevant to understand the complexity of more restrictive or more relaxed variants of the problem, for example by considering vertex colorings with more or fewer colors. Already in \cite{linial1987LOCAL}, Linial showed that an $O(\Delta^2)$-coloring can be computed deterministically in time only $O(\log^* n)$. Over the years, there were several papers that considered distributed coloring algorithms to color graphs with at least $\Delta+1$ colors~(e.g., \cite{barenboim10,BEK15,kuhn2020faster,ElkinPS15,SchneiderW10,kothapalli06,Maus21}). One however needs to use $\omega(\Delta)$ colors to obtain significantly faster distributed coloring algorithms. Colorings with less than $\Delta+1$ colors however require significantly more time. In \cite{linial1987LOCAL}, Linial shows that even on $\Delta$-regular trees, computing an $O(\sqrt{\Delta})$-coloring requires $\Omega(\log_\Delta n)$ rounds deterministically.\footnote{Linial uses an explicit construction of regular high-girth graphs with large chromatic number for his lower bound, but he remarks that by using the right probabilistic construction, the lower bound on the number of colors can be improved to $\Omega(\Delta/\log\Delta)$.} This was improved in \cite{brandt}, where it is shown that even computing a $\Delta$-coloring of $\Delta$-regular trees requires $\Omega(\log_\Delta n)$ rounds deterministically and $\Omega(\log_\Delta \log n)$ rounds with randomization. There are algorithms that nearly match those bounds~\cite{panconesi95delta,GhaffariHKM21,GhaffariKuhn21}. Further, by using network decompositions~\cite{awerbuch89,panconesi-srinivasan,linial93,elkin16_decomp,rozhonghaffari20}, it is possible to efficiently approximate the best possible vertex coloring~\cite{barenboim15_decomp}.\footnote{This approach exploits the standard distributed communication models, which in particular allow unbounded internal computations at all nodes.} In particular, in $\poly\log n$ rounds, it is possible to compute a coloring of a graph $G$ with $O\big(\frac{\log n}{\log\log n}\big)\cdot \chi(G)$ colors.

Another natural relaxation of the vertex coloring problem is the fractional coloring problem. A $c$-coloring of the nodes $V$ of a graph $G=(V,E)$ can be seen as a partition of $V$ into $c$ independent sets. A fractional $c$-coloring is an assignment of positive weights $\lambda_I$ to the independent sets $I$ of $G$ such that for every node $v\in V$, the total weight of the independent sets that contain $v$ is equal to (at least) $1$ and such that the total weight of all independent sets is equal to $c$. The smallest $c$ for which a graph $G$ has a fractional $c$-coloring is called the fractional chromatic number $\chi_f(G)$ of $G$. Alternatively, a fractional coloring can be defined as a multicoloring as follows. For two integer parameters $p$ and $q$ ($p\geq q$), a $(p:q)$-coloring is an assignment of (at least) $q$ colors to each node such that adjacent nodes are assigned disjoint sets of colors and such that the total number of distinct colors is equal to $p$. A $(p:q)$-coloring directly gives a fractional $(p/q)$-coloring and for any graph $G$ with fractional chromatic number $\chi_f(G)$, there exists a pair of integers $p$ and $q$ for which a $(p:q)$-coloring of $G$ with $p/q=\chi_f(G)$ exists.

From the distributed computing perspective, we believe that fractional colorings are interesting and relevant for two reasons. First, in some cases, where graph colorings are needed in practice, one can also use a fractional coloring. For example, if $G$ describes the possible conflicts between the wireless transmissions of nodes in a radio network, a $c$ coloring of the nodes of $G$ can be used to obtain a TDMA schedule for the nodes of $G$. The length of such a schedule is equal to the number of colors $c$ and every node can therefore be active in a $\frac{1}{c}$-fraction of all time slots. A $(p:q)$-coloring of the nodes of $G$ can also directly be used to obtain a TDMA schedule of length $p$ and in which every node is active in $q$ of the time slots (i.e., in all the time slots corresponding to its colors). Each node is therefore active in a $q/p$-fraction of all time slots. By computing a fractional coloring instead of a standard coloring, we can therefore potentially increase the fraction of active slots for each node and thus also the usage of the communication channel. Further, understanding the complexity of distributed fractional coloring will generally improve our understanding of the complexity of distributed coloring: what parts of the difficulty in computing vertex colorings stem from the fact that we need to assign exactly one color to every node (and that we thus need to break symmetries) and what parts of the difficulties remain if we compute fractional colorings, where we can ``average'' over a possibly larger total number of colors.

We are aware of three previous publications that studied the distributed fractional coloring problem. In \cite{stacs09}, it is shown that in any graph $G$, a fractional $(\mathrm{degree}+1)$-coloring with support $N!$ can be computed in a single deterministic communication round (where $N$ is the number of IDs). That is, there is a $1$-round algorithm that computes a multicoloring with $N!$ colors such that every node $v$ gets assigned a set of at least $N!/(\deg_G(v)+1)$ colors. That is, when considering fractional colorings, the standard $(\Delta+1)$-coloring problem can be solved in a single time step. It is further shown that a fractional $(1+\varepsilon)(\mathrm{degree}+1)$-coloring with support $O(\Delta^2\log N / \varepsilon^2)$ can also be computed in $1$ round. In both bounds, $N$ can be replaced by $C$ if an initial proper $C$-coloring of $G$ is given. Similar results were  also shown in \cite{HasemannHRS16}. Very recently, Bousquet, Esperet, and Pirot~\cite{BousquetEP21} made some interesting further progress on the distributed fractional coloring problem.

In \cite{BousquetEP21}, it is shown that although a fractional $(\Delta+1)$-coloring can be computed in a single communication round, in $\Delta$-regular graphs that do not contain $K_{\Delta+1}$ as a subgraph, the fractional $\Delta$-coloring problem requires time $\Omega(\log_\Delta n)$ deterministically and $\Omega(\log_\Delta\log n)$ rounds with randomization. That is, for the fractional $\Delta$-coloring problem, the same lower bounds as for the standard $\Delta$-coloring problem hold.\footnote{It has to be noted that for the standard $\Delta$-coloring problem, the lower bounds hold in $\Delta$-regular trees, but for the fractional $\Delta$-coloring problem, the known lower bounds (presented in \cite{BousquetEP21}) only hold for graphs in which every node is contained in some $K_{\Delta}$.} It is also shown that in graphs that do not contain $(\Delta+1)$-cliques, a $(q\Delta+1:q)$-coloring can be computed in time $O(q^3\Delta^{2q} + q\log^* n)$ deterministically. By setting $q=1/\eps$, this implies that for any $\eps>0$,  a fractional $(\Delta+\eps)$-coloring can be computed in time $O\big(\Delta^{O(1/\eps)}+ \frac{1}{\eps}\cdot\log^*n\big)$.
In addition, the paper shows that, for any constant $\epsilon>0$ and constant integer $d> 0$, in regular $d$-dimensional grids, it is possible to compute a fractional $(2+\epsilon)$-coloring in $O(\log^* n)$ rounds. Hence, while in bounded degree graphs, deterministically computing a fractional $\Delta$-coloring requires $\Omega(\log n)$ rounds, one can get arbitrarily close in only $O(\log^* n)$ rounds.  Moreover, in graphs from a minor-closed family of graphs and with sufficiently large  girth, it is possible to compute a fractional $(2+\varepsilon)$-coloring in $O\big(\frac{\log n}{\eps}\big)$ rounds, for any constant $\epsilon > 0$. Those results in particular imply that in some graphs, fractional colorings that are arbitrarily close to the best such colorings can be computed.

In this paper, we improve on the results of \cite{BousquetEP21} in several ways. We here give an overview over our results, for a detailed statement of the results, we refer to \Cref{sec:results}. First, we improve the time to compute a fractional $(\Delta+\eps)$-coloring. We show that a coloring of the same quality as in \cite{BousquetEP21} (i.e., a $(q\Delta+1:q)$-coloring for $q=1/\eps$) can be computed in time $O\big(\frac{1}{\eps^2}\cdot \sqrt{\Delta}\cdot\poly\log\Delta + \frac{1}{\eps}\cdot\log^* n\big)$ and a slightly worse $(q\Delta:q-1)$-coloring for $q=\Theta(\Delta/\eps)$ can be computed deterministically in $O\big(\log^2\big(\frac{1}{\eps}\big)\cdot\sqrt{\Delta}\poly\log\Delta + \big(1+\log_{\Delta}\frac{1}{\eps}\big)\cdot\log^* n\big)$ communication rounds. Moreover, if we further increase the total number of colors, a fractional $(\Delta+\eps)$-coloring can even be computed deterministically in time $O\big(\log^2\Delta +\log^2\big(\frac{1}{\eps}\big) + \log^3\big(\frac{1}{\eps}\big) / \log \Delta \big)$, i.e., we can improve the time dependency on $\Delta$ to polylogarithmic and we can drop the dependency on the number of nodes $n$ altogether. We further show that the dependency on $n$ can also be removed in the algorithm for fractionally coloring grids. In $d$-dimensional grids, for constant $\eps>0$ and constant $d$, a $(2+\eps)$ can be computed in $O(1)$ time. In addition, we study the problem of computing fractional colorings that are arbitrarily close to the best possible fractional colorings. For any $\eps>0$, we show that it is always possible to deterministically compute a fractional $(1+\eps)\cdot\chi_f(G)$-coloring of a graph $G$ in time $O\big(\frac{\log n}{\eps}\big)$ and we show that computing such a fractional coloring on trees requires $\Omega\big(\frac{\log n}{\eps}\big)$ rounds even with randomization. Note that this is in contrast to the standard coloring problem for which we are not aware of a $\poly\log n$-time algorithm that computes an approximation to the minimum vertex coloring problem with an approximation ratio that is better than $O\big(\frac{\log n}{\log\log n}\big)$. 

The remainder of this paper is organized as follows. In \Cref{sec:definitions} we describe the \LOCAL model of distributed computing and we give some useful definitions. \Cref{sec:results} contains the detailed statements of our contributions. In \Cref{sec:genericResults} we present some generic results that are then used in our algorithms. \Cref{sec:fastalgorithm} contains deterministic algorithms for computing a fractional $(\Delta + \varepsilon)$-coloring. In \Cref{sec:approximation} we show randomized and deterministic algorithms for computing arbitrarily good approximations of the chromatic number of a graph. Then, in \Cref{sec:lowerbound} we present a lower bound of $\Omega(\log n /\varepsilon)$ rounds for computing a fractional $(2+\varepsilon)$-coloring. Finally, \Cref{sec:grids} contains our constant-time algorithm for fractional $(2+\varepsilon)$-coloring on $d$-dimensional grids.

\section{Model and Definitions}
\label{sec:definitions}

\subsection{LOCAL model}
The model of computation that we consider is the well-known \LOCAL model of distributed computing. A distributed network is modeled as a graph where nodes are the computing entities, and edges represent communication links. Each node is equipped with a unique identifier (ID) from $\{1,\dotsc,n^c\}$ where $n$ is the total number of nodes in the graph and $c\ge 1$ is a constant. Initially, each node knows its own ID and degree, the maximum degree $\Delta$ of the graph, and the total number $n$ of the nodes. The computation proceeds in synchronous rounds, where at each round each node sends messages to its neighbors, receives messages from its neighbors, and performs some local computation. In the \LOCAL model the size of the messages and the local computation is not bounded. We say that an algorithm correctly solves a task in this model (e.g., a vertex coloring) in time $T$ if each node provides a local output (e.g., a color) within $T$ communication rounds, and the local outputs together yield a correct global solution (e.g., a proper coloring). In the randomized version of the \LOCAL model, additionally, each node is equipped with a random bit string. In this paper we will consider both Monte Carlo and Las Vegas randomized algorithms. A $T$-rounds Monte Carlo algorithm must always terminate within $T$ rounds, and the global output it produces must be correct with high probability, that is, with probability at least $1 - 1/n^c$, for an arbitrary constant $c \ge 1$. A $T$-rounds Las Vegas algorithm must terminate within $T$ rounds with high probability, and it must always produce a correct solution. Notice that, since the size of the messages is not bounded, we can see a $T$-round deterministic or a randomized Monte Carlo algorithm that runs at some node $u$ as a mapping of the $T$-hop neighborhood of $u$ into an output. This does not hold for Las Vegas algorithms, since there the running time is only bounded with high probability.

\subsection{Definitions}
We start by defining the notion of $(p:q)$-coloring. Informally, this coloring is an assignment of colors to the nodes of a graph, such that the colors come from a palette of $p$ colors, and such that to each node are assigned at least $q$ different colors. Neighboring nodes must have disjoint sets of colors.

\begin{definition}[$(p:q)$-coloring]
	Let $p\geq q\geq 1$. A $(p:q)$-coloring of a graph $G=(V,E)$ is an assignment of a set $X_v\subset [p]$ to each node $v \in V$ such that for all $v\in V$, $|X_v| \ge q$, and for all edges $\{u,v\}\in E$, $X_u\cap X_v=\emptyset$.
\end{definition}

Sometimes we are not interested in the total number of colors, but just in the ratio between the total number of colors and the number of colors assigned to the nodes. This notion is captured by the definition of fractional coloring.
\begin{definition}[Fractional $c$-coloring]
	A fractional $c$-coloring is a $(p:q)$-coloring satisfying $p/q \le c$.
\end{definition}
Given a $(p:q)$-coloring, we call $p$ the \emph{support} of the coloring. Naturally, apart from minimizing the ratio $p/q$ and the time for computing a fractional coloring, we also want to minimize the support $p$.

The minimum value $c$ for which there exists a fractional $c$-coloring is called fractional chromatic number.
\begin{definition}[Fractional chromatic number]
	The fractional chromatic number $\chi_f(G)$ of a graph $G$ is defined as
	\[
	\chi_f(G) := \inf\left\{\frac{p}{q}\,:\,G\text{ has a }(p:q)\text{-coloring}\right\} = \min\left\{\frac{p}{q}\,:\,G\text{ has a }(p:q)\text{-coloring}\right\}.
	\]
\end{definition}

\begin{definition}[Partial coloring]\label{def:partialcoloring}
	A partial $c$-coloring is a coloring of the vertices of a graph such that each node is either colored from a color in $\{1,\ldots,c\}$, or is uncolored. Similarly, a partial $(p:q)$-coloring is a coloring of the vertices of a graph such that each node, either has at least $q$ colors in $\{1,\ldots,p\}$, or is uncolored.
\end{definition}

We now provide some additional definitions that will be useful when describing our algorithms. Given a graph $G$, we denote with $\deg_G(v)$ the degree of node $v$ in $G$.
\begin{definition}[List coloring]
	In the $c_v$-list (vertex) coloring problem, each node $v$ is equipped with a list of arbitrary $c_v$ colors, and the goal is to assign to each node a color from its list, such that the resulting outcome is a proper coloring of the graph $G$. In particular, in the $(\mathrm{degree}+x)$-list coloring problem, each node $v$ has a list of size at least $\deg_G(v) + x$.
\end{definition}

\begin{definition}[Degree-choosability]
 A graph $G=(V,E)$ is \emph{degree-choosable} if it admits a $c_v$-list coloring for any list assignment satisfying $|c_v|\ge \deg_G(v), \forall v\in V$.
\end{definition}

\begin{definition}[Network decomposition]
	A $(c,d)$-network decomposition of the graph $G$ is a partition of the vertices of $G$ into at most $c$ disjoint color classes such that each connected subgraph induced by nodes of color $i$ has strong diameter at most $d$.
\end{definition}

\begin{definition}[Distance]
	Let $G=(V, E)$ be a graph. For a pair of nodes $u,v\in V$, we denote with $\ldist(u,v)$ the hop-distance between $u$ and $v$ in $G$. We also denote the distance between a node $v\in V$ and a set of nodes $S\subseteq V$ as $\ldist(v,S)=\min\{\ldist(v,u)~|~ u\in S \}$.
\end{definition}

\begin{definition}[Ruling set]
	An $(\alpha,\beta)$-ruling set is a set of nodes satisfying that the nodes in the set are at distance at least $\alpha$ between each other, and nodes not in the set are at distance at most $\beta$ from nodes in the set.
\end{definition}

\section{Our Results}
\label{sec:results}

Our first main contribution, presented in  \Cref{sec:fastalgorithm}, is in the improvement of the main algorithm of \cite{BousquetEP21}. In particular, we start by showing that a fractional  $(\Delta+\varepsilon)$-coloring, with small support, can be obtained in a time that depends only polynomially in $\Delta$ and $\varepsilon$. In \cite{BousquetEP21}, this dependency is exponential. In the following, we assume $\Delta \ge 3$ \footnote{We note that trivial adaptations of our algorithms also work for $\Delta=2$.}.
\begin{restatable}{theorem}{fastfractionalcoloring}\label{thm:fastfractionalcoloring}
	A $(q\Delta : q-1)$-coloring, for an arbitrary integer $q > 0$, can be deterministically computed in time $O(\alpha^2\log\Delta \cdot T + \alpha\log^*n)$ in the \LOCAL model, where $T$ is the time required to solve the $(\mathrm{degree}+1)$-list coloring problem given an $O(\Delta^2)$-coloring in input, and where $\alpha=O(1+\log_\Delta q)$.
\end{restatable}
If we set $q = \Theta(\Delta/\varepsilon)$, we obtain the following corollary.
\begin{restatable}{corollary}{corfastfractionalcoloring}
	For any $\epsilon>0$, the fractional $(\Delta+\varepsilon)$-coloring problem, with support $O(\Delta^2/\varepsilon)$, can be solved deterministically in time
	\[
	O\left(\left(\log\Delta + \frac{\log^2(1/\varepsilon)}{\log\Delta}\right)\cdot T + \left(1 + \frac{\log(1/\varepsilon)}{\log\Delta}\right)\cdot\log^* n\right),
	\]
 	where $T$ is the time required to solve the $(\mathrm{degree}+1)$-list coloring problem given an $O(\Delta^2)$-coloring in input.
\end{restatable}
Since the $(\mathrm{degree}+1)$-list coloring problem, given an $O(\Delta^2)$-coloring, can be solved in $O(\sqrt{\Delta} \poly \log \Delta)$ rounds deterministically \cite{fraigniaud16,BEG18,MausTonoyan20}, we obtain the following corollary. 
\begin{corollary}
	For any constant $\epsilon>0$, the fractional $(\Delta+\varepsilon)$-coloring problem, with support $O(\Delta^2)$, can be solved in $O(\sqrt{\Delta} \poly \log \Delta  + \log^* n)$ deterministic rounds.
\end{corollary}

Then, we show that we can obtain a different tradeoff between the support and the running time.
\begin{restatable}{theorem}{smallsupport}\label{thm:smallsupport}
	A $(q\Delta +1 : q)$-coloring, for an arbitrary integer $q > 0$, can be deterministically computed in time $O(q^2\log\Delta \cdot T + q\log^*n)$ in the \LOCAL model, where $T$ is the time required to solve the $(\mathrm{degree}+1)$-list coloring problem given an $O(\Delta^2)$-coloring in input.
\end{restatable}
If we set $q = \Theta(1/\varepsilon)$, since $T = O(\sqrt{\Delta} \poly\log \Delta)$, we obtain the following corollary, that shows that at the cost of a slightly worse running time, we obtain a better support. 
\begin{restatable}{corollary}{corsmallsupport}
	For any $\epsilon>0$, the fractional $(\Delta+\varepsilon)$-coloring problem, with support $O(\Delta/\varepsilon)$, can be solved deterministically in time
	\[
	O\left(\frac{1}{\eps^2}\cdot \sqrt{\Delta}\cdot\poly\log\Delta + \frac{1}{\eps}\cdot\log^* n\right).
	\]
\end{restatable}

We then prove that, at the cost of drastically increasing the number of colors, it is possible to improve the dependency on $\Delta$, and to entirely remove the dependency on $n$.
\begin{restatable}{theorem}{veryfast}\label{thm:veryfast}
	For any $\epsilon> 0$, the fractional $(\Delta+\varepsilon)$-coloring problem can be solved deterministically in time
	\[
          O\left(\left(\log\Delta + \frac{\log^2(1/\varepsilon)}{\log\Delta}\right)\cdot \log (\Delta/\varepsilon) \right) =
          O\left(\log^2\Delta + \log^2 \frac{1}{\varepsilon} + \frac{\log^3 (1/\varepsilon)}{\log \Delta}\right).
	\]
\end{restatable}
\begin{corollary}
	For any $\epsilon> 1 / \Delta^c$, where $c > 0$ is an arbitrary constant, the fractional $(\Delta+\varepsilon)$-coloring problem can be solved in $O(\log^2 \Delta)$ deterministic rounds.
\end{corollary}

Our second contribution, presented in \Cref{sec:approximation}, is in showing that, by allowing a logarithmic dependency on $n$ in the running time, we can obtain fractional colorings that are arbitrarily close to the optimum. We provide both randomized and deterministic algorithms, that differ in the required support and in the running time. Let $p / q = \chi_f(G)$ be the fractional chromatic number of $G$. The algorithms that we provide do not require to know $p$ and $q$, but if these values are provided, or even if just the value of $\chi_f(G)$ is provided, then our algorithms obtain a fractional coloring with smaller support.
In particular, if $p$ and $q$ are known to the nodes, let $p'=p$ and $q' = q$. Otherwise, let $p' = \chi c \log n / \epsilon^2$ and $q' = (1-\varepsilon)p' / \chi_f(G)$, where $\chi = \chi_f(G)$ if $\chi_f(G)$ is known to the nodes, and $\Delta+1$ otherwise. We first show the following.
\begin{restatable}{theorem}{randomizedapx}\label{thm:randomizedapx}
	Let $G=(V,E)$ be a graph that admits a $(p:q)$ coloring, and let $t=O(\log n/\varepsilon)$, for an arbitrary $\epsilon > 0$. There is a randomized \LOCAL algorithm that, with high probability, computes a $(tp':(1-\varepsilon)tq')$-coloring, that is, a fractional $(1+O(\varepsilon))\frac{p}{q}$-coloring, in $O(\log n / \varepsilon)$ rounds.
\end{restatable}
We then show two different deterministic algorithms, that provide different tradeoffs between the support and the running time.
\begin{restatable}{theorem}{detapxA}\label{thm:detapx1}
	Let $G=(V,E)$ be a graph that admits a $(p:q)$-coloring, and let $t=O(\poly n / \varepsilon)$, for an arbitrary $\epsilon > 0$. There is a deterministic \LOCAL algorithm that computes a $(t p':(1-\varepsilon)t q')$-coloring, that is, a fractional $(1+O(\varepsilon))\frac{p}{q}$-coloring, in $O(\log n / \varepsilon)$ rounds.
\end{restatable}
\begin{restatable}{theorem}{detapxB}\label{thm:detapx2}
	Let $G=(V,E)$ be a graph that admits a $(p:q)$ coloring, and let $t=O(\log n/\varepsilon)$, for an arbitrary $\epsilon > 0$. There is a deterministic \LOCAL algorithm that computes a $(tp':(1-\varepsilon)tq')$-coloring, that is, a fractional $(1+O(\varepsilon))\frac{p}{q}$-coloring, in $O(\log n (\log^2 n + \mathrm{ND})/ \varepsilon)$ rounds, where $\mathrm{ND}\le \poly \log n$ is the time required to compute an $(O(\log n), O(\log n))$-network decomposition.
\end{restatable}

We then prove, in \Cref{sec:lowerbound}, that the $O(\log n / \varepsilon)$ time dependency for computing an almost optimal fractional coloring is necessary, even on trees, and even for randomized algorithms.
\begin{restatable}{theorem}{lowerbound}\label{thm:lowerbound}
	Computing a fractional $(2+\epsilon)$-coloring on trees in the \LOCAL model requires $\Omega(\log n / \varepsilon)$, even for randomized algorithms.
\end{restatable}

Finally, in \Cref{sec:grids}, we move our attention to grids. In \cite{BousquetEP21}, it has been shown that, for any constant $\epsilon$ and $d$, in $d$-dimensional grids, it is possible to compute a fractional $(2+\varepsilon)$-coloring in time $O(\log^* n)$.  We show that the same problem can be solved in constant time.
\begin{restatable}{theorem}{griddet}\label{thm:grid-det}
	Let $G$ be a $d$-dimensional grid. For any $\varepsilon > 0$, there is a deterministic \LOCAL algorithm that computes a fractional $(2 + \varepsilon)$-coloring on $G$, that runs in $2^{O(d^2 + d \log \frac{1}{\varepsilon})}$ rounds.
\end{restatable}

\section{Generic results}
\label{sec:genericResults}
In this section, we prove some generic statements that will be useful in the following sections.

\subsection{From partial colorings to fractional colorings}
 We start by showing that, given an algorithm that computes a partial fractional coloring satisfying that each node has some fixed probability of being colored, then it is possible, at the cost of increasing the total number of colors, to compute a proper fractional coloring.
\begin{lemma}\label{thm:generic}
	Assume there exists a randomized algorithm $A(n,\varepsilon)$ that runs in $T(n,\varepsilon)$ rounds and, with probability at least $1 - f$, such that $f = f(n,\varepsilon) \le \varepsilon$, computes a partial $(p : q)$-coloring satisfying that each node is uncolored with probability at most $\epsilon$. Then, there exists a randomized algorithm that runs in $T(n,\varepsilon/4)$-rounds and, for an arbitrary $f'$, with probability at least $1 - f'$, computes a $(p' : q')$-coloring, where $p' = p t$, $q' = (1 - \varepsilon) q t$, and $t = O(\frac{1}{\varepsilon}\log\frac{n}{f'})$. 
\end{lemma}
\begin{proof}
	We run $A(n , \varepsilon / 4)$ for $t = \frac{6}{\varepsilon}\log \frac{n}{f'}$ times in parallel. Clearly, the running time is still $T(n,\varepsilon/4)$. For each execution, we use an independent palette of colors of size $p$, and hence we obtain a palette of $pt$ total colors. We need to prove that, with probability at least $1 - f'$, we assign at least $(1 - \varepsilon) q t$ colors to each node.
	
	Consider an arbitrary node $u$. Let $X_i=1$ if node $u$ is uncolored during execution $i$, and $X_i=0$ otherwise. By assumption, a node is uncolored with probability at most $\varepsilon / 4 + f(n,\varepsilon / 4) \le \varepsilon/2$, hence $P(X_i=1)\le\varepsilon/2$. Let $X=\sum_{i=1}^{t}X_i$. By linearity of expectation, we get that $\mathbb{E}[X]\le\varepsilon t / 2$. By a Chernoff bound, we get that: 
	\[
	P(X\ge \varepsilon t)\le e^{\frac{-\varepsilon t}{6}}\le \frac{f'}{n}.
	\]
	By a union bound we get that each node is uncolored in at most $\varepsilon t$ colorings with probability at least $1 - f'$. Hence, with probability at least $1 - f'$, each node receives $q$ colors for at least $(1-\varepsilon)t$ times.
\end{proof}

\subsection{From private randomness to shared randomness}
We now prove a useful lemma, that allows us to reduce the number of random bits used by a randomized algorithm. We will later use this lemma to derandomize fractional coloring algorithms without increasing the support too much. Note that, in the statement of the lemma, the number of bits used by the original algorithm does not play a role in the resulting algorithm. In fact, as a starting point, we essentially only need to know that the amount of randomness, as a function of $n$, can be bounded.
\begin{lemma}\label{lem:lessbits}
	Let $A$ be a randomized algorithm that runs in $T$ rounds and solves some problem $P$ with probability of success at least $1 - f$, by using $b = b(n)$ bits of private randomness, where nodes have no private inputs except for their random bit strings and their identifiers. Then, there exists a randomized algorithm $A'$ that uses only $O(\log \frac{n}{f})$ bits of shared randomness and solves $P$ in $T$ rounds with probability of success at least $1 - 2f$.  
\end{lemma}
\begin{proof}
	The proof follows a standard argument along the lines of Newman's theorem in communication complexity (see, e.g., \cite{kushilevitz_nisan_1996}).
	Let $B = N b$, where $N = n^c$, for some constant $c\ge 1$, is the size of the ID space. Note that any algorithm that uses $b$ bits of private randomness can be trivially converted into an algorithm that uses $B$ bits of shared randomness. Also, note that by fixing the string of shared randomness used by the nodes, we obtain a deterministic algorithm.
	
	Let $\ell = 2^B$, and let $R = \{R_1, \ldots, R_\ell\}$ be the set of possible shared random bit string assignments. We will prove below, using the probabilistic method, that there is a set $S \subseteq R$ of size $k = O(\frac{n^2}{f})$ satisfying that, for any graph $G$ of $n$ nodes, algorithm $A$ fails for at most a $2f$ fraction of the strings in $S$. This means that we can construct an algorithm $A'$ that, given a bit string of shared randomness of length $O(\log k ) = O(\log \frac{n}{f})$, chooses an element from $S$ uniformly at random, and then executes the $T$-round algorithm $A$ by using that bit string. This algorithm $A'$ runs in $T$ rounds and solves $P$ with a failure probability of at most $2f$.
	
	We now prove that, by choosing $S$ at random, there is non-zero probability that it satisfies the requirements.
	We construct $S = \{s_1, \ldots , s_k\}$ by sampling with replacement $k$ strings from $R$. Consider an arbitrary graph $G$ of $n$ nodes. Let $X_i=1$ if the execution of $A$ on $G$ by using the bit string $s_i$ would fail. Otherwise, let $X_i=0$. Since $s_i$ has been chosen uniformly at random, and since the failure probability of $A$ is at most $f$, then $P(X_i=1) \le f$.
	By linearity of expectation, it holds that $\mathbb{E}[X]\le f k$. Let $X = \sum_{1 \le i \le k} X_i$. Note that the variables are clearly independent. Hence, by a Chernoff bound, we get that 
	\[
	P(X \ge 2 f k) \le e^{-\frac{f k }{3}}.
	\]
	If $X \ge 2 f k$ we say that $S$ is \emph{bad} for $G$. Note that there are at most $2^{n^2}$ graphs of $n$ nodes with no private inputs, and that there are at most $\binom{N}{n}$ possible ID assignments on each graph. Hence, by a union bound, the probability that $S$ is bad for \emph{some} graph of $n$ nodes is at most 
	\[
		\binom{N}{n} \cdot 2^{n^2} \cdot  e^{-\frac{f k }{3}} \le e^{c n \ln n + n^2 - \frac{f k}{3}}.
	\]
	Choosing $k = 6n^2/f$ ensures that the above expression is strictly less than $1$, for $n$ sufficiently large. Note that $k = O(\frac{n^2}{f})$.
\end{proof}

\subsection{From randomized fractional colorings to deterministic fractional colorings}
We now show that, given a randomized algorithm for fractional coloring, it is possible to obtain a deterministic algorithm with the same running time, at the cost of increasing the support.
\begin{lemma}\label{thm:genericB}
	Assume there exists a randomized $T$-round algorithm $A$ that computes a $(p : q)$-coloring with probability at least $1 - f$. Then, there exists a deterministic $T$-round algorithm that computes a $(p' : q')$-coloring, where $p' = p t$, $q' = (1-2f)q t$, and $t = 2^{O(\log \frac{n}{f})}$.
\end{lemma}
\begin{proof}
	We first apply \Cref{lem:lessbits} to reduce the amount of randomness required by algorithm $A$ to $B = O(\log \frac{n}{f})$ bits of shared randomness, obtaining a new algorithm $A'$ that runs in $T$ rounds and computes a $(p : q)$ coloring with probability at least $1 - 2f$.  We cannot directly run $A'$: since nodes are not provided with shared randomness, this is not possible. Instead, we run $A'$ in parallel for all possible $t = 2^B$ values of the shared random bit string assignment. In each run, we use an independent palette of $p$ colors, and hence we use $p t$ colors in total. Since algorithm $A'$ succeeds with probability at least $1 - 2f$, then in at least a fraction $1-2f$ of the executions all nodes receive at least $q$ colors. Hence, we obtain that each node receives at least $(1-2f)q t$ colors in total.
\end{proof}

\section{Fast algorithm}\label{sec:fastalgorithm}
In this section, we present an algorithm that is able to compute a fractional $(\Delta+\varepsilon)$-coloring, with a running time that depends only polynomially in $\Delta$, and that requires small support. Later, we will show how to modify the algorithm to obtain a logarithmic dependency in $\Delta$, at the cost of having a much larger support. More formally, we start by proving the following theorem.
\fastfractionalcoloring*

\paragraph*{High level idea}
Our algorithm, on a high level, works as follows. We first compute a clustering of the graph by computing an $(\alpha,\beta)$-ruling set, for suitable parameters $\alpha$ and $\beta$, and by connecting each node to the cluster centered at the nearest ruling set node. We then proceed by coloring the nodes in a special way. In particular, we compute $q$ different colorings, such that each coloring is a $\Delta$-coloring of the graph, where some nodes are allowed to be uncolored, and such that each node is uncolored in at most one of the $q$ colorings.

In order to prove that such a coloring can be computed efficiently, we will exploit an important property of the computed clustering: for each cluster, it holds that it either contains many nodes (at least $q$), or it contains a degree-choosable component. We will show that this implies that we can color the nodes such that, in each cluster that contains a degree-choosable component, the $\Delta$-coloring problem is solved properly, while in the other clusters we can choose a specific node to leave uncolored. If the cluster is large enough, we can have a different uncolored node for each of the $q$ colorings, obtaining that each node is uncolored for at most one coloring.

The computed coloring allows us to assign at least $q-1$ colors to each node of the graph, where the colors come from a palette of size $q\Delta$, and hence we obtain a $(q\Delta : q-1)$-coloring. By choosing the right size of the clusters, we can also prove that the running time is polynomial in $\Delta$, and more precisely that it is only slightly worse than the time required to solve the $(\mathrm{degree}+1)$-list coloring problem, which we use as a subroutine.

\subsection{The clustering}
We start by proving that it is possible to compute a clustering of a graph in such a way that it satisfies some desirable properties. In particular, we prove the following lemma.
\begin{lemma}\label{lem:clustering}
	Let $\alpha = c(1+\log_\Delta q)$, for some constant $c$ and an arbitrary integer $q>0$. 
	It is possible to compute a clustering of a graph $G$ such that the following holds.
	\begin{itemize}
		\item Each cluster has a strong diameter of at most $2 \alpha^2 \log \Delta$.
		\item Each cluster contains: 
		\begin{itemize}
			\item at least $q$ nodes, or 
			\item a node of degree at most $\Delta - 1$, or
			\item a degree-choosable component, such that all neighbors of the nodes in the degree-choosable component are also contained in the cluster.
		\end{itemize}

	\end{itemize}
	This clustering can be computed in $O(\alpha^2 \log \Delta + \alpha \log^* n)$ deterministic rounds.
\end{lemma}
In order to prove this lemma, we will use the following lemmas present in the literature.
\begin{lemma}[Lemma 8 of \cite{GhaffariHKM21}]\label{lem:choosable-or-expansion}
	Let $G=(V,E)$ be a graph and $v \in V$ be a node such that inside the $r$-radius neighborhood
	of $v$ there are no degree-choosable components and every node has degree $\Delta$. Then, for each even $r$,
	there are at least $(\Delta - 1)^{r/2}$ nodes at distance $r$ from $v$.
\end{lemma}

\begin{lemma}[Ruling sets, \cite{Balliu0O20,SchneiderEW13}]\label{lem:ruling-sets}
	A $(2,\beta)$-ruling set can be computed in time $O(\beta \Delta^{2/\beta} + \log^* n)$ deterministic rounds.
\end{lemma}
We are now ready to prove \Cref{lem:clustering}.
\begin{proof}
	We start by computing an $(\alpha, (\alpha-1)\alpha \log \Delta)$-ruling set, by computing a $(2,\alpha \log \Delta)$-ruling set on $G^{\alpha-1}$, the $(\alpha-1)$th power of $G$. By applying \Cref{lem:ruling-sets} with $\beta = \alpha \log \Delta$, this ruling set can be computed in $T = O(\alpha^2 \log \Delta + \alpha \log^* n)$ deterministic rounds. Then, in additional $O(\alpha^2 \log \Delta)$ rounds, each node finds the nearest ruling set node (breaking ties arbitrarily), and joins its cluster.
	
	Clearly, the running time requirement is satisfied. We need to prove that this clustering satisfies the required properties. The first property is clearly satisfied, since $\alpha^2 \log \Delta$ is an upper bound on the distance between an arbitrary node and its nearest ruling set node.
	
	We now show that the second property is also satisfied. Let $v$ be the ruling set node of an arbitrary cluster $C$, that is, its center. By the definition of $(\alpha,\beta)$-ruling set, all nodes at distance at most $\lfloor \alpha /2 -1 \rfloor$ from $v$ are all contained in $C$. Consider the set $S$ of nodes at distance at most $k$ from $v$, where $k \in \{\lfloor \alpha /2 -2 \rfloor, \lfloor \alpha /2 -3 \rfloor  \}$ is even. If there is a node of degree at most $\Delta-1$ in $S$, then the property is clearly satisfied. Otherwise, all nodes in $S$ have degree $\Delta$, and by \Cref{lem:choosable-or-expansion} we get that either there is a degree-choosable component in $S$, or there are at least $(\Delta-1)^{k/2}$ nodes in $S$. In the first case, note that the neighbors of nodes in $S$ are all contained in $C$, and hence the property is satisfied. In the second case, we obtain that the cluster contains at least $(\Delta-1)^{\lfloor c(1+\log_\Delta q) /2 -3 \rfloor/2}$ nodes, that, for a large enough constant value of $c$, is at least $q$, and hence the property is satisfied in this case as well.	
\end{proof}

\subsection{The algorithm}
We now describe an algorithm that computes a fractional $(\Delta+\varepsilon)$-coloring. In the following, we will make use of the notion of partial $\Delta$-coloring (see \Cref{def:partialcoloring}).

On a high level, the main idea of the algorithm is to compute $q$ (possibly) different partial $\Delta$-colorings (where each of the colorings come from a different palette), such that each node is uncolored in at most $1$ of the colorings. In this way, we can assign $q-1$ colors to each node, from a palette of $q\Delta$ total colors. We now show how the properties stated in \Cref{lem:clustering} can be used for this purpose.
\begin{lemma}\label{lem:onecolor}
	Let $G$ be a graph that is clustered according to \Cref{lem:clustering}, where each cluster that contains at least $q$ nodes also contains a special marked node. Let $T$ be the time required to solve the $(\mathrm{degree}+1)$-list coloring problem, and let $R$ be the bound on the diameter of the clusters. Then, in $O(T \cdot R)$ deterministic rounds it is possible to solve the partial $\Delta$-coloring problem such that only marked nodes remain uncolored.
\end{lemma}
\begin{proof}
	We prove this lemma by slightly modifying the core of the deterministic $\Delta$-coloring algorithm presented in Ghaffari et al.\ \cite{GhaffariHKM21}. Each node $v$ starts by spending $R$ rounds to gather its entire cluster $C_v$. In this way, it can see if there is a marked node, or a node with degree at most $\Delta - 1$, or a degree-choosable component (at least one of these cases must apply).
	If there is a marked node $z$ in $C_v$, let $S_{C_v} = \{z\}$, otherwise, if there is a node $z'$ with degree at most $\Delta - 1$, then let $S_{C_v} = \{z'\}$, otherwise, let $S_{C_v}$ be the set of nodes of an arbitrary degree-choosable component guaranteed to exist by \Cref{lem:clustering}. Note that, without additional communication, $v$ can compute its distance from the nearest node $s$ in $S_{C_v}$.
	
	Let $G_i$ be the subgraph induced by nodes at distance $i$ from their nearest node in the set, that is, $G_i = \{ u ~|~ \ldist(u,S_{C_u}) = i \}$. Let $G_{>i}$ (resp.\ $G_{\ge i}$) be the graph induced by all nodes contained in $G_j$, for all $j>i$ (resp.\ $j\ge i$). Note that, for all $i > 0$, the nodes of $G_i$, in $G_{\ge i}$, have degree at most $\Delta-1$, since the maximum degree in $G$ is $\Delta$, and all nodes in $G_i$ have at least one neighbor in $G_{i-1}$. Also, note that $G_{>R}$ is empty. 
	
	We proceed as follows. Assume that $G_{>i}$ is properly $\Delta$-colored (we start with $i=R$, where the statement trivially holds, since $G_{> R}$ is empty), and let $c(u)$ be the color of a node $u$ in $G_{> i}$. We show that we can $\Delta$-color $G_{\ge i}$.
	For each node $v$ in $G_i$, we define its list of available colors as $L_v = \{1,\ldots, \Delta\} \setminus \{c(u) ~|~ u \in N(v) \cap G_{>i}\}$. Since the degree of nodes of $G_i$ in $G_{\ge i}$ is at most $\Delta-1$, then $L_v$ defines a $(\mathrm{degree}+1)$-list coloring instance of $G_i$, that can be solved in $T$ rounds. 
	By iterating this procedure for $i=R,\ldots,1$, we obtain that all nodes of $G$, except the ones in $S=\bigcup \{S_{C_v} ~|~ v\in V\}$, are properly $\Delta$-colored. 
	
	Finally, we handle the nodes in $S$. If $S_{C_v}$ contains a marked node, we just leave it uncolored. Otherwise, if $S_{C_v}$ contains a node with degree at most $\Delta - 1$, we color it with an arbitrary available color. Otherwise, if $S_{C_v}$ contains a degree-choosable component, then, for each node $u \in  S_{C_v}$, we define $L_u$ as above. This time, $L_u$ defines a $\mathrm{degree}$-list coloring instance. Note that, in general, the $\mathrm{degree}$-list coloring problem may be unsolvable, but this is never the case in a degree-choosable component, by definition. Since, for each pair of clusters $C_1, C_2$, $S_{C_1}$ and $S_{C_2}$ are non-adjacent, then it is possible to solve all the $\mathrm{degree}$-list coloring instances in parallel, by brute force, in $R$ rounds. 
\end{proof}

We are now ready to describe the algorithm. First of all, nodes start by computing an $O(\Delta^2)$-coloring.
Then, nodes compute the clustering described in \Cref{lem:clustering}. Let $R = O(\alpha^2 \log \Delta)$ be the maximum diameter of the clusters. Then, nodes spend $R$ rounds to check the type of their cluster, that is, if there is a degree-choosable component satisfying the required property, or if there are at least $q$ nodes, or if there is a node with degree at most $\Delta - 1$. In all clusters $C$ containing at least $q$ nodes, we choose $q$ arbitrary distinct nodes $\{v_{C,1},\ldots,v_{C,q}\}$.
Then, we apply \Cref{lem:onecolor} for $q$ times in parallel. During the application $i$, the nodes that are considered marked are $\{v_{C,i}\}$. We obtain $C_1,\dots,C_q$ partial $\Delta$-colorings of $G$, such that each node is uncolored in at most one coloring. Hence, we use a palette of $q\Delta$ colors, such that each node has at least $q-1$ colors, that is, we obtain a $(q\Delta : q-1)$-coloring.

\paragraph*{Time complexity}
The previously described algorithm computes a $(q\Delta : q-1)$-coloring. Hence, in order to prove \Cref{thm:fastfractionalcoloring}, we need to give a bound on its running time.
Computing the $O(\Delta^2)$-coloring can be done in $O(\log^* n)$ rounds. Computing the clustering requires $O(\alpha^2 \log \Delta + \alpha \log^* n)$ rounds. Gathering the cluster requires $O(\alpha^2 \log \Delta)$ rounds. The application of \Cref{lem:onecolor} requires $O(T \cdot \alpha^2 \log \Delta)$ rounds, where $T$ is the time for solving a $(\mathrm{degree}+1)$-list coloring instance given an $O(\Delta^2)$-coloring. Recall that $\alpha = c(1+\log_\Delta q)$. Hence, we obtain an overall time complexity of $O(\alpha^2 \log \Delta \cdot T + \alpha \log^* n)$, where $\alpha = O(1+\log_\Delta q)$.

\subsection{Faster algorithm}
We now show that, at the cost of drastically increasing the number of colors, we can improve the dependency on $\Delta$, and entirely remove the dependency on $n$. In particular, we will prove the following.
\veryfast*

We start by explaining how to remove the $O(\alpha \log^* n)$ dependency from the algorithm of \Cref{thm:fastfractionalcoloring}, obtaining the following intermediate result.
\begin{lemma}\label{lem:veryfastintermediate}
	For any $\varepsilon > 0$, the fractional $(\Delta+\varepsilon)$-coloring problem can be solved deterministically in time
	\[
	O\left(\left(\log\Delta + \frac{\log^2(1/\varepsilon)}{\log\Delta}\right)\cdot T\right),
	\]
	where $T$ is the time required to solve the $(\mathrm{degree}+1)$-list coloring problem given an $O(\Delta^2)$-coloring in input.
\end{lemma}
\begin{proof}
 On a high level, the $\log^* n$ dependency appears in the time complexity of the algorithm for two reasons:
\begin{itemize}
	\item The algorithm spends $O(\alpha \log^* n)$ rounds for computing the clustering, and this is because, in order to compute an $(\alpha,\beta)$-ruling set, we first need to find an $O(\Delta^{2\alpha})$-coloring of $G^\alpha$.
	\item The algorithm spends $O(\log^* n)$ rounds to find an $O(\Delta^2)$ coloring of $G$. Note that, given an $O(\Delta^{2\alpha})$-coloring of $G^\alpha$, we can reduce this runtime to $O(\log^* \Delta^{2\alpha})$.
\end{itemize}
 Hence, if, in $G^\alpha$, nodes are already provided with an $O(\Delta^{2\alpha})$-coloring, then we can get rid of the $\log^* n$ dependency, obtaining a faster algorithm, that requires only $O(\alpha^2 \log \Delta \cdot T + \log^* \Delta^{2\alpha}) =O(\alpha^2 \log \Delta \cdot T)$ rounds. Unfortunately, finding such a coloring requires $\Omega(\log^* n)$ rounds \cite{linial92}. In order to overcome this issue, we do the following. We first spend a constant number of rounds to find a color randomly, such that each node has a large enough probability of being colored. Then, we run \Cref{thm:fastfractionalcoloring} on the subgraph induced by nodes that are colored correctly. In this way, we obtain an algorithm that computes a partial fractional coloring, satisfying that each node is uncolored with some fixed probability. We can then apply \Cref{thm:generic} and \Cref{thm:genericB} to obtain a deterministic algorithm that finds a proper fractional coloring.

In order to find the required coloring, we proceed as follows. Each node picks a color uniformly at random from a palette of $c$ colors, for some value $c$ to be fixed later. Then, each node $v$ checks its $\alpha$-radius neighborhood, and if there is a node $u$ with the same color, then $v$ becomes uncolored. For any node $u$ contained in the $\alpha$-radius neighborhood of $v$, we have that $P(u \text{ and } v \text{ pick different colors}) \ge 1 - 1/c$. Let $\bar{\Delta} = \Delta^\alpha$ be an upper bound on the degree of $G^\alpha$. Then, $P(u \text{ is colored}) \ge (1 -  1/ c)^{\bar{\Delta}}$, that by the Bernoulli inequality is at least $1 - \bar{\Delta} / c$.
Hence, by applying the algorithm of \Cref{thm:fastfractionalcoloring} on the subgraph induced by colored nodes, we obtain an algorithm that runs in $O(\alpha^2 \log \Delta \cdot T)$ rounds and computes a partial $(q\Delta : q-1)$-coloring where each node is uncolored with probability at most $\bar{\Delta} / c$. 
We can now apply \Cref{thm:generic} (note that, since our algorithm never fails, then $f=0$), obtaining a randomized algorithm that terminates in $O(\alpha^2 \log \Delta \cdot T)$ rounds and computes a $(q\Delta t : (1 - \bar{\Delta} / c) (q-1)t)$-coloring, for some $t$ that depends on the target failure probability $f'$ and $n$. Then, we can apply \Cref{thm:genericB} to obtain a deterministic algorithm that runs in $O(\alpha^2 \log \Delta \cdot T)$ rounds and computes a $(q\Delta t' : (1-2f')(1 - \bar{\Delta} / c)(q-1)t')$-coloring, for some $t'$ that depends on $f'$ and $n$. By fixing $c = q\bar{\Delta}$ and $f' = \frac{1}{2q}$, we obtain a deterministic algorithm for computing a fractional $\frac{q\Delta}{(1-1/q)^2 (q-1)}$-coloring, that, for $q = O(\Delta / \varepsilon)$, is a fractional $(\Delta + O(\varepsilon))$-coloring. The running time, for such a value of $q$, is $O((\log \Delta + \log^2(1/\varepsilon)/\log \Delta)\cdot T)$, as required. 
\end{proof}

\paragraph*{Proof of \Cref{thm:veryfast}}

We now explain how to remove the dependency on $T$, and obtain \Cref{thm:veryfast}. Let $q = \Theta(\Delta/ \varepsilon)$. In order to obtain a faster algorithm, we start by showing that we can replace the dependency on $T$ in \Cref{lem:onecolor}, with $O(\log q)$, at the cost of leaving some nodes uncolored. Hence, we obtain a ``sloppy'' version of \Cref{lem:onecolor}. Recall that $T$ is the time required to run a procedure that solves a $(\mathrm{degree}+1)$-list coloring instance.
Instead of running such a procedure, we run a procedure that partially solves an instance in $O(\log q) = O(\log \frac{\Delta}{\varepsilon})$ rounds, such that a node remains uncolored with probability at most $1/q$.
This can be achieved by letting each node try to pick an available color uniformly at random for $O(\log q)$ times \cite{Johansson99}.

We now show that it is possible to obtain an algorithm that computes a partial $(\Delta : 1)$-coloring satisfying that each node is uncolored with probability at most $2/q$. First, we compute the clustering as in the original algorithm, that is, by applying \Cref{lem:clustering}. Then, on clusters of size at least $q$, we mark a node of the cluster chosen uniformly at random. By applying the sloppy version of \Cref{lem:onecolor}, we obtain what we need.

Finally, by applying \Cref{thm:generic} and \Cref{thm:genericB}, we obtain a deterministic algorithm that solves fractional $(\Delta + O(\varepsilon))$-coloring in time $O((\log\Delta + \frac{\log^2(1/\varepsilon)}{\log\Delta})\cdot \log (\Delta/\varepsilon) )$, proving the theorem.

\subsection{Better support}
We now provide a different algorithm, that has slightly worse running time compared to the one of  \Cref{thm:fastfractionalcoloring}, but that is able to give a fractional $(\Delta+\varepsilon)$-coloring with smaller support. In particular, we will prove the following theorem.
\smallsupport*
For this purpose, we will exploit the following lemma, first presented in \cite{AubryGT14}, and used also in \cite{BousquetEP21}.
\begin{lemma}[Proposition 8 of \cite{AubryGT14}]\label{lem:pathcoloring}
	Let $q \ge 1$ be an integer and let $P = (v_1, \ldots, v_{2q+1})$ be a path. Assume that for $i \in \{1,2q+1\}$ the vertex $v_i$ has a list $L_{v_i}$ of at least $q+1$ colors, and for any $2 \le i \le 2q$, $v_i$ has a list $L_{v_i}$ of at least $2q+1$ colors. Then, each vertex $v_i$ of $P$ can be assigned a subset $S_i \subseteq L_{v_i}$ of $q$ colors, so that adjacent vertices are assigned disjoint sets.
\end{lemma}

Similarly as in the algorithm of \Cref{thm:fastfractionalcoloring}, we start by computing an $(\alpha,(\alpha -1)\alpha \log \Delta)$-ruling set, by computing a $(2,\alpha \log \Delta)$-ruling set on $G^{\alpha-1}$, the $(\alpha-1)$-th power of $G$. This time, we choose $\alpha = 4q+4$. Then, we also compute an $O(\Delta^2)$ coloring. Then, we form clusters by letting each node $u$ join the cluster $C_v$ centered at the nearest ruling set node $v$. By the definition of $(\alpha,\beta)$-ruling set, all nodes at distance at most $\lfloor \alpha/2 - 1 \rfloor$ from $v$ are contained in $C_v$. Hence, in $C_v$, there exists at least one induced path of $\lfloor \alpha/2 - 1 \rfloor$  nodes satisfying that all neighbors of nodes in $P$ are fully contained in $C_v$ (that is, the path obtained by taking $v$ and the nodes contained in some shortest path from $v$ to some node at distance $\lfloor \alpha/2 - 2 \rfloor$ from $v$). Note that $\lfloor \alpha/2 - 1 \rfloor = 2q+1$. Let $S_v$ be the set of nodes of the path.

Similarly as in the proof of \Cref{lem:onecolor}, we can spend $O(\alpha^2 \log \Delta \cdot T)$ rounds to find a partial $\Delta$-coloring that leaves uncolored only nodes in $\bigcup S_v$. This partial coloring can be trivially converted into a partial $(q \Delta : q)$-coloring where each colored node has exactly $q$ colors. Then, for each node $u$ contained in some $S_v$, we can define its list of available colors as $\{1,\ldots,q\Delta+1\} \setminus \{C(z) ~|~ z \in N(u)\}$, where $C(z)$ is the set of colors assigned to node $z$. Note that, from the list of available colors of $u$, we removed at most $q$ colors for each neighbor of $u$. Hence, the endpoints of the path satisfy $|L_u| \ge q+1$ since they have at most $\Delta - 1$ colored neighbors, and the inner nodes satisfy $|L_u| \ge 2q+1$ since they have at most $\Delta - 2$ colored neighbors. Hence, by applying \Cref{lem:pathcoloring} we get that nodes can complete the coloring by brute force, since paths are not connected to each other.

We spend $O(q^2 \log \Delta + q \log^* n)$ rounds to compute a ruling set, $O(\log^* n)$ rounds to compute the $O(\Delta^2)$ coloring. The rest of the algorithm requires $O(q^2 \log \Delta \cdot T)$ rounds. Each node receives $q$ colors, from a palette of total $q\Delta+1$ colors. Hence, the theorem follows.

\section{Approximating the fractional chromatic number}\label{sec:approximation}
In this section, we show that it is possible to find arbitrarily good approximations of the fractional chromatic number. In particular, given a graph $G$ with fractional chromatic number $\chi_f(G) = p/q$, we provide randomized and deterministic algorithms that are able to find a fractional $(1+\varepsilon)\chi_f(G)$-coloring, for any $\varepsilon>0$. Our algorithms use a different amount of total colors, depending on whether the nodes know $p$ and $q$ or not, or if they know $\chi_f(G)$.

On a high level, we show that it is possible to cluster a graph such that each node is unclustered with probability at most $\varepsilon$, and such that any pair of clusters is at least 2 hops apart (that is, for any two nodes that are in different clusters it holds that they do not share an edge). In this way, inside each cluster, we can optimally solve the fractional $\chi_f(G)$-coloring, or find a good-enough approximation if $p$ and $q$ are not known. Then, by applying \Cref{thm:generic} and  \Cref{thm:genericB}, we obtain randomized and deterministic algorithms for computing a fractional coloring that approximates the fractional chromatic number.   

\subsection{Computing a clustering}
In order to obtain a clustering of the graph, we slightly modify the clustering algorithm of Miller, Peng, and Xu~\cite{miller2013parallel} (MPX), to make it compute clusters that are $2$ hops apart, such that each node is unclustered with probability at most $\varepsilon$, and each cluster has weak diameter $O(\log n / \varepsilon)$.
\begin{lemma}\label{lem:MPXmod}
	Given a graph $G=(V,E)$, there is a randomized algorithm that computes a $2$-hops-apart clustering of $G$ such that each node is unclustered with probability at most $\varepsilon$, and each cluster has weak diameter $O(\log n / \varepsilon)$ with high probability. This algorithm terminates in $O(\log n / \varepsilon)$ rounds with high probability.
\end{lemma}
\begin{proof}
	Since we are going to use a modified version of the MPX procedure, we start by describing the standard MPX procedure.
	
	Each node $u$ chooses independently a \emph{shift} $\delta_u$ from an exponential distribution with parameter $\gamma=\varepsilon/2$. Let the \emph{shifting distance} from $u$ to $v$ be denoted as $\ldist\__\delta(u,v)=\ldist(u,v)-\delta_u$. Each node $v$ is then assigned to the cluster $C_u$ centered at the node $u$ that, among all nodes in $G$, minimizes the value of the shifted distance $\ldist\__\delta(u,v)$, breaking ties arbitrarily. Miller, Peng, and Xu~\cite{miller2013parallel} proved that, with high probability, each cluster has radius $O(\log n /\varepsilon)$. While the original procedure is designed to work in the PRAM model, it is folklore that it can be easily converted into a distributed algorithm that terminates in $O(\log n /\varepsilon)$ rounds with high probability. This procedure obtains a partitioning of the vertices into clusters satisfying the following properties.
	\begin{itemize}
		\item Each cluster is connected, that is, for any two nodes $u$ and $v$ it holds that, if both are in the same cluster $C$, then the nodes in the shortest path between $u$ and $v$ are also in $C$.
		\item Each cluster has strong diameter $O(\log n /\varepsilon)$ with high probability.
		\item Each edge has probability at most $\varepsilon$ to be an intercluster edge.
	\end{itemize}
	We run this procedure, and then we modify the obtained clustering in the following way.\footnote{A similar modification has been used, for example, in \cite{FaourKuhn20}.} Let $\{u,v\}\in E$ be an edge such that its endpoints $u$ and $v$ are on different clusters. For each such edge, we choose the endpoint with the smallest identifier and we remove it from the cluster. We say that these nodes are \emph{unclustered}. After this operation it holds that, for any two nodes that are in different clusters, they do not share an edge, meaning that the clusters are at least $2$ hops apart, as desired. Notice that, by removing nodes from the clusters we may lose the guarantee on the strong diameter of the clusters. However, it still holds that the \emph{weak} diameter of each cluster is $O(\log n /\varepsilon)$ w.h.p., and this is enough for our purposes.
	
	What is left to show is that each node is unclustered with probability at most $\varepsilon$. Consider an unclustered node $u$. Let $C_{u'}$ be the cluster centered at node $u'$ where $u$ belonged to before being removed. Since $u$ is unclustered, it means that there is a node $w$ neighbor of $u$ such that, before removing nodes from the clusters, nodes $u$ and $w$ were assigned to different clusters (otherwise $u$ would not be unclustered). Let $C_{w'}$ be the cluster centered at $w'$ that was initially assigned to node $w$. By construction of these clusters, we have that $\ldist\__\delta(u',u)\le \ldist\__\delta(w',w) + 1$ and, at the same time, $\ldist\__\delta(w',w)\le \ldist\__\delta(u',u) + 1$. Hence, $|\ldist\__\delta(u',u)- \ldist\__\delta(w',w)|\le 1$, implying that  $|\ldist\__\delta(u',u)- \ldist\__\delta(w',u)|\le 2$. In other words, the absolute value of the difference between the smallest and the second smallest shifting distance of an unclustered node is at most $2$. In \cite{miller2013parallel} it has been proven that, for each node, the probability that this event happens is at most $2\gamma=\varepsilon$. We thus obtain that each node is unclustered with probability at most $\varepsilon$.
\end{proof}
In the remaining of this section we use the variables $p'$ and $q'$, that are defined as follows.
\begin{itemize}
	\item If $p$ and $q$ are known to the nodes, then $p'=p$ and $q' = q$.
	\item Otherwise, let $p' = \chi c \log n / \epsilon^2$ and $q' = (1-\varepsilon)p' / \chi_f(G)$, where $\chi = \chi_f(G)$ if $\chi_f(G)$ is known to the nodes, and $\Delta+1$ otherwise.
\end{itemize}

\subsection{Solving a partial fractional coloring}
We now prove that, by using the clustering algorithm of \Cref{lem:MPXmod}, it is possible to find a partial $(p' : q')$-coloring.

\begin{lemma}\label{lem:clustercoloring}
	There exists a randomized $O(\log n / \varepsilon)$-round algorithm $A$ that computes a partial $(p' : q')$-coloring satisfying that, with probability at least $1-1/n$, each node is uncolored with probability at most $\epsilon$.
\end{lemma}
\begin{proof}
	Note that the algorithm described in \Cref{lem:MPXmod} is Las Vegas, but we can turn it into a Monte Carlo algorithm by truncating its execution after $O(\log n / \varepsilon)$ steps, and leave unclustered every node that did not terminate. Since the original algorithm terminates in $O(\log n / \varepsilon)$ rounds with high probability (that is, at least $1-1/n$), then this new algorithm always terminates in $O(\log n / \varepsilon)$  rounds, which is also an upper bound on the diameter of the clusters, and leaves each node unclustered with probability at most $\varepsilon + 1/n$ by a union bound. Hence, by slightly scaling $\varepsilon$, we obtain the same guarantees as the original algorithm.

	Hence, we start by running the (Monte Carlo variant of the) clustering algorithm.
	Then, since each cluster has weak diameter at most $R = O(\log n / \varepsilon)$, we can spend $R$ rounds for computing in parallel, in each cluster, by brute force, a $(p' : q')$-coloring (we will later argue why such a coloring always exists). Note that this is possible even if $q'$ is not known to the nodes, as they can just find the best possible solution. Since unclustered nodes do not get a color, and since clusters are $2$ hops apart, then there are no neighboring nodes that get the same color.
	
	We need to argue why a $(p' : q')$-coloring always exists. Let $G$ be a graph that is $(p : q)$-colored. We show that there is a randomized process that, with non-zero probability, produces a $(p' : q')$-coloring. By the probabilistic method, this implies that $G$ admits a $(p' : q')$-coloring, and hence that also each cluster admits a $(p' : q')$-coloring.
	
	We sample with replacement $p'$ colors from $p$ colors. Consider an arbitrary node $u$. Let $X_i=1$ if, during the $i$-th sampling, we sample a color that node $u$ has among its colors. Otherwise, let $X_i=0$. Since node $u$ has at least $q$ out of the $p$ possible colors, then $P(X_i=1)\ge q/p$. Let $X=\sum_{i=1}^{p'}X_i$. By linearity of expectation, it holds that $\mathbb{E}[X]\ge \frac{p'q}{p}=\frac{p'}{\chi_f(G)} = q' / (1-\varepsilon)$. By a Chernoff bound, we get that 
	
	\[
	P(X\le q') \le e^{-\frac{\varepsilon^2}{2}\cdot\frac{p'}{\chi_f(G)}}
	= e^{-\frac{\varepsilon^2 c \chi \log n}{\varepsilon^2 2 \chi_f(G)}}
	\le n^{-c/2}.
	\]	
	By a union bound, we have that each node has less than $q'$ colors with probability at most $n^{1 - c/2}$. By choosing $c\ge 4$ we get that each node has at least $q'$ colors with probability at least $1-1/n$.
\end{proof}

\subsection{Putting things together}
We now prove the existence of randomized and deterministic algorithms for approximating the fractional chromatic number, with running time $O(\log n / \varepsilon)$. Note that, for $\varepsilon < 1/n$, we can trivially solve the problem by gathering the entire graph on a node and brute forcing a solution. Hence, in the following, assume $\varepsilon \ge 1/n$.

\Cref{lem:clustercoloring} guarantees the existence of a randomized algorithm that runs in $O(\log n / \varepsilon)$ rounds, and with probability at least $1 - f$, where $f = 1/n$, computes a partial $(p' : q')$-coloring satisfying that each node is uncolored with probability at most $\varepsilon$. By applying \Cref{thm:generic}, we obtain that there exists a randomized algorithm, that also runs in $O(\log n / \varepsilon)$ rounds and, with probability at least $1-f'$, where $f' = 1/n^c$ for an arbitrary constant $c \ge 1$, computes a partial $(p'' : q'')$-coloring, where $p'' = p' t$, $q'' = (1- \varepsilon)q' t$, and $t = O(\log n / \varepsilon)$. Hence, we obtain the following.
\randomizedapx*

Then, starting from this algorithm, we can apply \Cref{thm:genericB}, where $f = 1/n^c$ for an arbitrary constant $c\ge1$, and obtain that there exists a deterministic algorithm that also runs in $O(\log n / \varepsilon)$ rounds, and computes a partial $(p''' : q''')$-coloring, where $p''' = p'' t'$, $q''' = (1- 2f)q'' t'$, and $t' = \poly n$. Since $\varepsilon \ge 1/n \ge f$, then this means that, compared to the randomized algorithm, we lose at most an $O(\varepsilon)$ fraction of colors. Hence, we obtain the following.
\detapxA*

\subsection{Less colors}
We now show an alternative way for derandomizing the algorithm of \Cref{thm:randomizedapx}, obtaining the same guarantees on the number of colors, but with a slightly worse running time. For this purpose, we use the following result of Ghaffari, Harris, and Kuhn~\cite{ghaffari2018derandomizing}.
\begin{theorem}[Theorem 1.1 of \cite{ghaffari2018derandomizing}]\label{thm:slocalderandomization}
	Any $r$-round randomized \LOCAL algorithm for a locally checkable problem can be transformed into a deterministic \LOCAL algorithm with complexity $O(r (\log^2 n + \mathrm{ND}))$, where $\mathrm{ND}$ is the time required to compute an $(O(\log n), O(\log n))$-network decomposition.
\end{theorem}
We note that Theorem 1.1 of \cite{ghaffari2018derandomizing} applies to randomized algorithms that satisfy the following requirements.
\begin{enumerate}
	\item They always terminate within $r$ rounds.
	\item Each node sets a flag to $1$ if its output is incorrect, and to $0$ otherwise. With high probability, the flag should be $0$.
	\item Each node should be able to check, in $O(1)$ rounds, whether its flag has the correct value.
\end{enumerate}
In order to apply \Cref{thm:slocalderandomization} and derandomize the algorithm of \Cref{thm:randomizedapx}, we need to show that we can tweak the algorithm to satisfy the above requirements. The algorithm presented in \Cref{thm:randomizedapx} 
clearly satisfies the first requirement. However, it is not compatible with the second one, since the guarantee on the quality of the coloring does not always hold. More precisely, the number of colors obtained by a node only holds with high probability, and for this reason, if $q'$ is not known, a node may not notice that its output is incorrect, and it fails to set its flag correctly. Let us see how to handle this issue.

If $q'$ is known to the nodes, then they can just set their flag to $1$ if they have strictly less than $(1-\varepsilon)q't$ colors, and $0$ otherwise. If $q'$ is not known, we perform a preprocessing step to compute a value of $q'$ (that may be different for different nodes), that will then be used by the nodes to decide whether to set their flag or not.
The preprocessing step works as follows.  Let $T$ be the running time of the algorithm that we want to derandomize. Each node $v$ spends $2T$ rounds to gather its $2T$-radius neighborhood. Then, it checks, in that neighborhood, what is the maximum value $q'$ for which there exists a $(p't : (1-\varepsilon)q't)$-coloring. Observe that the obtained value $q'$ is a lower bound on the number of colors that $v$, while executing the algorithm of \Cref{thm:randomizedapx}, is able to obtain when it belongs to a cluster. This follows from the fact that, any cluster where $v$ belongs to must be fully contained in its $2T$-radius neighborhood. Note that the preprocessing step also guarantees that the third requirement of Theorem 1.1 of \cite{ghaffari2018derandomizing} is satisfied, since each node can just check whether the flag is set correctly depending on the number of colors that it has and the value of $q'$.

Hence, by combining \Cref{thm:slocalderandomization} with the obtained variant of the algorithm of \Cref{thm:randomizedapx}, we obtain the following theorem.
\detapxB*

\section{Lower bound}\label{sec:lowerbound}
In this section, we show a lower bound of $\Omega(\log n / \varepsilon)$ rounds for computing a fractional $(2+\varepsilon)$-coloring, which holds already on trees, and even for randomized algorithms. We first show that there exist graphs with girth $\Omega(\log n/\varepsilon)$ and fractional chromatic number strictly larger than $2 + \varepsilon$. Then we argue that, if we take an algorithm that, on trees, achieves a fractional $(2+\varepsilon)$-coloring in $o(\log n/\varepsilon)$ rounds, and we execute it on such graphs, it must fail, since the obtained fractional coloring would be too good.
We get our contradiction on the existence of such an algorithm by arguing that, in time $o(\log n/\varepsilon)$, a node cannot distinguish whether it is on a tree or on these graphs.

We start by describing a graph family of interest. Let $\mathcal{G^*}$ be a graph family that contains all graphs with $n$ nodes, $m=O(n)$ edges, girth $\Omega(\log n)$, and where the largest independent set has size at most $n/c$, for some large constant $c$, and for infinite values of $n$. Such a family of graphs is known to exist \cite{LubotzkyPS88}. Starting from a graph $G=(V,E)\in\mathcal{G^*}$, we construct a graph $H$ by replacing all edges of $G$ with a path of length $2k+1$, for some $k=\Theta(1/\varepsilon)$. In other words, let $e=\{u,v\}$ be an edge in $G$. We replace $e$ by a path $P^e=(u,w^e_1,w^e_2,\dotsc,w^e_{2k},v)$. We refer to nodes in $P^e$ as \emph{path nodes}, and we refer to the $w^e_1,w^e_2,\dotsc,w^e_{2k}$ nodes as \emph{inner path nodes}. Let $\mathcal{G}$ be the family that contains all such graphs. By construction, a graph $H\in\mathcal{G}$ has $N=n+2km$ nodes and girth $\Omega(\log n / \varepsilon)$. We now show the following lemma about the fractional chromatic number of these graphs.

\begin{lemma}
	Let $H$ be a graph in $\mathcal{G}$. The fractional chromatic number of $H$ is strictly larger than $2 /(1-c'\varepsilon)$, for some constant $c'>0$.
\end{lemma}
\begin{proof}
	Let $S$ be any independent set of $H$. We modify $S$ and compute a new independent set $S'$ of $H$ such that $|S'|\ge|S|$ and $S'$ contains exactly $mk$ inner path nodes. Note that this implies that for each path $P^e=(u,w^e_1,w^e_2,\dotsc,w^e_{2k},v)$, at most one of the two endpoints can be in $S'$. 
	Let $P^e=(u,w^e_1,w^e_2,\dotsc,w^e_{2k},v)$ be a path in $H$ where at most $k-1$ inner path nodes are in $S$. There are two cases: either at most one of the two endpoints is in $S$, or both $u$ and $v$ are in $S$. 
	
	In the former case, we modify $S$ in the following way. W.l.o.g., let $u$ be a node not in $S$. We start from $u$ and compute an optimal independent set inside $P^e$ sequentially, by starting with $w^e_1$ in the independent set and then putting every other node in the set. This procedure puts in the independent set $k$ inner path nodes of $P^e$, and note that the obtained set is still independent.
	
	In the latter case, we remove node $u$ from the set, and then we proceed as in the former case. Note that the obtained set is still independent, and it is not smaller. In fact, before the modification, there were at most $k - 1$ nodes of $P^e \setminus \{u, v\}$ in $S$, and hence at most $k$ nodes of $P^e \setminus \{v\}$ in $S$, and after the modification we have $k$ nodes of $P^e \setminus \{u,v\}$ in the set.
	
	We call $S'$ the independent set resulting from the above operations. Notice that $|S'|\ge|S|$. Recall that each graph in $\mathcal{G}$ is obtained from a graph $G \in \mathcal{G}^*$ where the largest independent set has size at most $n/c$. Note also that, by construction, if we project $S'$ onto $G$, then we obtain and independent set of $G$.
	
	Hence, 
	\begin{align*}
		|S'| &\le \frac{n}{c} + mk\\
		& =\frac{N}{2} - \frac{(c-2)n}{2c}\\
		& = \frac{N}{2}\left(1 - \frac{2(c-2)n}{2cN}\right)\\
		& = \frac{N}{2}\left(1 - \frac{(c-2)n}{c(n + 2km)}\right)\\
		& = \frac{N}{2}\left(1 - \frac{(c-2)n}{c(n + 2\frac{c_1}{\varepsilon} c_2n)}\right) \mbox{ where } k=\frac{c_1}{\varepsilon} \mbox{ and } m=c_2n\\
		& = \frac{N}{2}\left(1 - \frac{c-2}{c(1 + 2\frac{c_1}{\varepsilon} c_2)}\right)\\
		& < \frac{N}{2}\left(1 - c_3\varepsilon\right), \mbox{ for some constant $c_3 > 0$ that depends on $c_1$ and $c$.}
	\end{align*}
	Hence, any color can be assigned to strictly less than a fraction $(1 - c_3 \varepsilon)/2$ of the nodes, implying that the fractional chromatic number of $H$ is strictly larger than $2/ ( 1 - c_3\varepsilon)$.
\end{proof}
From the above lemma we get the following corollary.

\begin{corollary}\label{cor:lb-graphs}
	There exists an infinite family of graphs of girth $\Omega(\log n/\varepsilon)$ and fractional chromatic number strictly larger than $2+\epsilon$.
\end{corollary}
We are now ready to prove our main theorem.

\begin{theorem}
	Computing a $(2+\epsilon)$-fractional coloring on trees in the \LOCAL model requires $\Omega(\log n / \varepsilon)$, even for randomized algorithms.
\end{theorem}
\begin{proof}
	The proof follows a standard indistinguishability argument, already used to prove, e.g., that coloring trees with $o(\Delta / \log \Delta)$ colors requires $\Omega(\log_\Delta n)$ rounds, even for randomized algorithms \cite{linial92}. For simplicity we prove our claim for the deterministic case, but it can be extended to the randomized case with standard techniques.
	
	Suppose there exists an algorithm $A_T$ that, on trees, computes a fractional $(2 + \varepsilon)$-coloring in $o(\log n / \varepsilon)$ rounds. Let $G$ be a graph satisfying \Cref{cor:lb-graphs}. In $o(\log n/\varepsilon)$ rounds, each node in $G$ does not see any cycle, and hence we could run algorithm $A_T$ on $G$ and it would not notice that it is being run on a graph that is not a tree. However, $A_T$ must fail on $G$, since, by \Cref{cor:lb-graphs}, the fractional chromatic number of $G$ is strictly larger than $2 + \varepsilon$. Hence, suppose we run $A_T$ on $G$ and it fails on the neighboring nodes $u$ and $v$ who, as an output of $A_T$, got some common color (note that the algorithm may also fail on just one node by assigning to it too few colors, but the lower bound argument follows in the same way).
	
	Recall that a $t$-round algorithm in the \LOCAL model can be seen as a mapping from $t$-radius neighborhoods into outputs.
	Let $B_u$ and $B_v$ be the views of radius $t$ of nodes $u$ and $v$, respectively, that $A_T$ then maps into the outputs of $u$ and $v$. Let $B=B_u\cup B_v$. The subgraph in $B$ does not contain cycles and it has radius $o(\log n/\varepsilon)$.
	Starting from $B$, we construct a tree that contains $B$, and where we add additional nodes in order to obtain an $n$-node tree $T$ (note that the additional nodes can be added without altering the $t$-radius views of $u$ and $v$, since, in $B$, there exists at least one node at distance at least $t$ from both of them). Nodes $u$ and $v$ in $T$ have the same exact view as in $G$, hence they output the same improper fractional coloring, meaning that $A_T$ fails on $T$, which is a contradiction. Hence, $A_T$ cannot exist, proving the theorem.
\end{proof}

\section{Grids}\label{sec:grids}
In \cite{BousquetEP21}, it has been shown that, for any constant $\epsilon$ and $d$, in $d$-dimensional grids, it is possible to compute a fractional $(2+\varepsilon)$-coloring in time $O(\log^* n)$.  We show that the same problem can be solved in constant time. 

The algorithm of \cite{BousquetEP21} computes a $(2q+4 \cdot 6^d : q)$-coloring that runs in $O(d \ell (2\ell)^d + d \ell \log^* n)$ rounds, where $\ell = q + 2 \cdot 6^d$. The running time is dominated by the time required to compute a maximal independent set on $G^{\ell}$, where the distance is taken w.r.t. the infinity norm (that is, for two nodes $u$ and $v$ with coordinates $(u_1,\ldots,u_d)$ and $(v_1,\ldots,v_d)$ their distance is $\max_{1\le i \le d}\{|u_i - v_i|\}$). In fact, a maximal independent set can be computed in time $O(\Delta + \log^* n)$, and on $G^{\ell}$ we have that $\Delta = (2\ell +1)^{d}$, and hence there is an overhead of $O(d \ell)$ in the runtime. After computing the independent set, the rest of the algorithm requires just $O(d\ell(2\ell)^d)$ rounds. By setting $q = 2^{O(d + \log \frac{1}{\varepsilon})}$, this algorithm gives a fractional $(2 + \varepsilon)$-coloring in time $ T_{\mathrm{coloring,n}} + T_{\mathrm{rest}}$, where $T_{\mathrm{coloring}, c} = 2^{O(d + \log \frac{1}{\varepsilon})}\log^* c$ and  $T_{\mathrm{rest}} = 2^{O(d^2 + d \log \frac{1}{\varepsilon})}$.
Similarly as in the proof of \Cref{lem:veryfastintermediate}, we can replace the $\log^* n$ dependency with $\log^* c$, if nodes are provided with a distance-$d \ell$ $c$-coloring.

We compute a partial distance-$d \ell$ coloring by letting nodes pick a color uniformly at random. Nodes that obtain an invalid coloring, uncolor themselves. We would like to execute the algorithm of \cite{BousquetEP21} on the subgraph induced by colored nodes, but we cannot, since the subgraph is not a grid anymore. In order to solve this issue, we consider only nodes satisfying that, within their running time, they cannot notice that the graph is not a grid. We call these nodes \emph{happy}. In other words, a node is happy if and only if, within the running time of the algorithm, it does not see any uncolored node. On these nodes, by a standard indistinguishability argument, the algorithm must work correctly. Note that this running time depends on $c$, but we will pick a value of $c$ satisfying that the total running time is anyways strictly less than $k T_{\mathrm{rest}}$, for some large enough constant $k$. The probability that a node is colored is at least $1 - \Delta/c$. Since a node sees at most $d^{kT_{\mathrm{rest}}}$ nodes within its running time, the probability that a node is happy is at least $1 - \frac{\Delta}{c} d^{k T_{\mathrm{rest}}}$, meaning that a node is unhappy with probability at most $\frac{\Delta}{c} d^{k T_{\mathrm{rest}}}$. We want this probability to be at most $\varepsilon$, and for that, we can pick $c = \Delta d^{k T_{\mathrm{rest}}} / \varepsilon$. Note that, for such a value of $c$, $T_{\mathrm{coloring}, c} + T_{\mathrm{rest}} \le k T_{\mathrm{rest}}$, as required. 

Hence, there is an algorithm that in $2^{O(d^2 + d \log \frac{1}{\varepsilon})}$ rounds computes a partial fractional $(2+\varepsilon)$-coloring satisfying that each node is uncolored with probability at most $\varepsilon$. By applying \Cref{thm:generic} and \Cref{thm:genericB}, we obtain the following.
\griddet*

%%
%% Bibliography
%%
%\bibliographystyle{abbrv}
\bibliography{references}
\newpage
\appendix

\end{document}